\documentclass[runningheads]{llncs}
\usepackage{amsmath}
\usepackage{amssymb}
\usepackage{graphicx}
\usepackage{graphics}
\usepackage{listings}
\usepackage{multirow}
\usepackage{framed}

\title{Chromatic and flow polynomials of generalized vertex join graphs and outerplanar graphs}

\titlerunning{Graph polynomials of generalized vertex joins}

\author{Boris Brimkov and Illya V. Hicks}

\authorrunning{B. Brimkov and I.V. Hicks}

\institute{
Computational \& Applied Mathematics, Rice University, Houston, TX 77005, USA\\
\email{boris.brimkov@rice.edu, ivhicks@rice.edu}
}

\begin{document}
\maketitle

\begin{abstract}
A generalized vertex join of a graph is obtained by joining an arbitrary multiset of its vertices to a new vertex. We present a low-order polynomial time algorithm for finding the chromatic polynomials of generalized vertex joins of trees, 
and by duality we find the flow polynomials of arbitrary outerplanar graphs. 

We also present closed formulas for the chromatic and flow polynomials of vertex joins of cliques and cycles, otherwise known as ``generalized wheel" graphs.

\medskip

\end{abstract}

\section{Introduction}

Graph polynomials contain various information about graphs and are used to analyze the properties of graphs and networks. Two of the most important single-variable graph polynomials are the chromatic polynomial and the flow polynomial. Studying these polynomials is an active area of research with many theoretical consequences and practical applications. 

The chromatic polynomial is related to the Potts model, which models the behavior of certain elements of solid state physics, like crystals and ferromagnets. The coefficients of the chromatic polynomial have another application in statistical mechanics -- they are used to model state changes, as described in the work of Sokal \cite{application}. The chromatic polynomial is also related to Stirling numbers, and thus finds applications in a variety of analytic and combinatorics problems; see \cite{stirling2,stirling1}.

There are important theoretical questions surrounding the flow polynomial as well, such as the 5-flow conjecture and other flow existence theorems and conjectures. The flow polynomial also has an application in crystallography and statistical mechanics as it is related to models of the physical properties of ice and crystal lattices \cite{ice1,ice2}. For more applications of chromatic and flow polynomials, see the comprehensive survey of Ellis-Monaghan and Merino \cite{tuttemain} and the bibliography therein.

Unfortunately, computing the chromatic and flow polynomials of a graph are very challenging tasks. These problems are NP-hard for general graphs, and even for bipartite planar graphs as shown in \cite{nphard}, and sparse graphs with $|E|=\mathcal{O}(|V|)$. In fact, most of the terms of the chromatic and flow polynomials of general graphs cannot even be approximated (see \cite{inapprox,nphard}).

Thus, a large volume of work in this area is focused on exploiting the structure of specific types of graphs in order to derive closed formulas, algorithms, or heuristics for computing their chromatic and flow polynomials. In particular, classes of graphs which are generalizations of trees, cliques, and cycles are frequently investigated. 

For example, Wakelin et al. \cite{polytree2,polytree1} consider a class of graphs called polygon trees and find their chromatic polynomials; they also characterize the chromatic polynomials of biconnected outerplanar graphs and the flow polynomials of their dual graphs.  Whitehead \cite{surveychrom,twotrees} characterizes the chromatic polynomials of a class of clique-like graphs called $q$-trees. 
Furthermore, Lazuka \cite{cactus} obtains explicit formulas for the chromatic polynomials of cactus graphs, Gordon \cite{tuttetrees} studies Tutte polynomials (a generalization of chromatic and flow polynomials) of rooted trees, and Mphako-Banda \cite{flower,tutteflower} derives formulas for the chromatic, flow, and Tutte polynomials of flower graphs.

In this paper, we consider yet another generalization of trees, cliques, and cycles. We define a generalized vertex join of a graph $G$ to be the graph obtained by joining an arbitrary multiset of the vertices of $G$ to a new vertex. We find the chromatic polynomials of generalized vertex joins of trees, cliques, and cycles, and use the duality of chromatic and flow polynomials to find the flow polynomials of certain other classes of graphs, including outerplanar graphs. Thus, we complement the work of Wakelin et al. \cite{polytree2,polytree1} on chromatic polynomials of outerplanar graphs and flow polynomials of their duals, by characterizing the flow polynomials of outerplanar graphs and the chromatic polynomials of their duals. Several related results are included as well.

The paper is organized as follows. In the next section, we recall some notions and notations related to graph theory and graph polynomials. In Section~3, we list well-known technical tools used in the computation of chromatic and flow polynomials. In Section~4, we calculate the chromatic polynomials of generalized vertex joins of trees; we relate these results to outerplanar graphs in Section~5. In Section~6, we consider generalized vertex joins of cliques and cycles, and related dual results. We conclude with some final remarks and open questions in Section~7.

\section{Preliminaries}

We assume the reader is familiar with basic graph theoretic notions and operations. For an extensive background on graph theory, one is referred to \cite{bondy}. In this section, we first recall the definition of a multiset and related terms, 
followed by select graph theoretic notions used in the paper.

Let $G=(V,E)$ be a graph. A \emph{multiset} $S$ over $V$ is a collection of vertices of $V$, each of which may appear more than once in $S$. The number of times a vertex $v$ appears in $S$ is the \emph{multiplicity} of $v$. The \emph{underlying set} of $S$ is the set $S'$ which contains the (unique) elements of $S$. For example, if $V=\{v_1,v_2,v_3,v_4,v_5\}$, a multiset $S$ over $V$ could be $\{v_1, v_1, v_3,v_4, v_4, v_4\}$ and the underlying set of $S$ is $S'=\{v_1,v_3,v_4\}$. Using this notion, we define the \emph{generalized vertex join of $G$ using $S$} to be the graph $G_S=(V\cup \{v^*\}, E\cup \{vv^*:v\in S\})$. Note that if the multiplicity of $v$ in $S$ is $m$, there are $m$ parallel edges between $v$ and $v^*$ in $G_S$. See Figure 1 for an example.

\begin{figure}
\begin{center}
\includegraphics[scale=0.4]{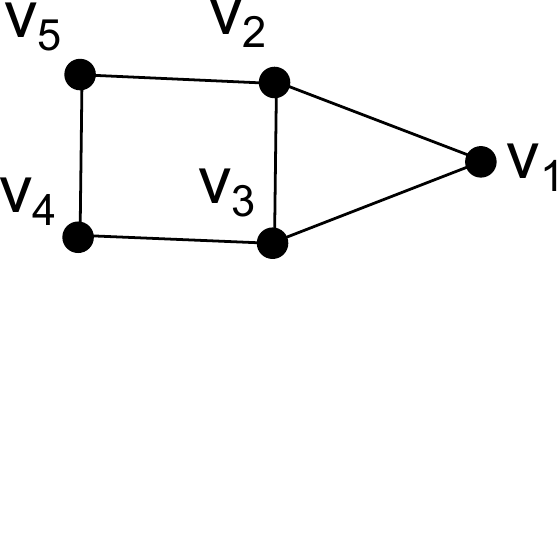}\qquad \qquad
\includegraphics[scale=0.4]{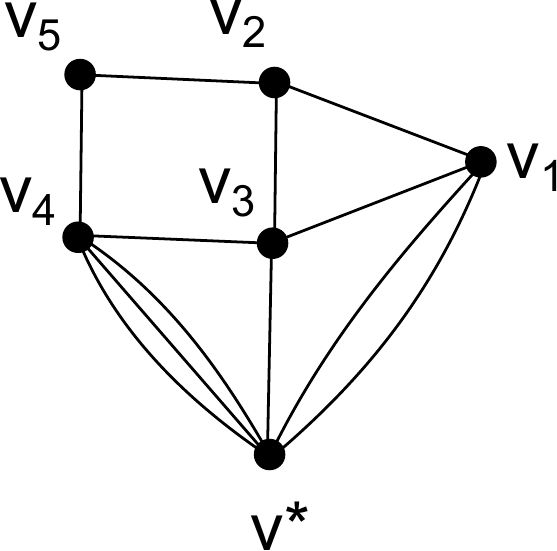}\qquad \qquad
\caption{\emph{Left:} A graph $G$. \emph{Right:} The generalized vertex join of $G$ using $S=\{v_1, v_1, v_3,v_4, v_4, v_4\}$.}
\end{center}
\end{figure}

Given $G=(V,E)$ and $S\subset V$, the \emph{induced subgraph} $G[S]$ is the subgraph of $G$ whose vertex set is $S$ and whose edge set consists of all edges of $G$ which have both ends in $S$. Given $u,v\in V$, the \emph{contraction} $G/uv$ is obtained by identifying $u$ and $v$ and removing any edges between them. Note that $G$ does not need to have the edge $uv$ for $G/uv$ to be defined. Finally, we say that $G$ is \emph{biconnected} if $G-v$ has exactly one connected component for all $v\in V$.




Many of the graphs considered in this paper are \emph{planar} graphs --- i.e., they can be drawn in the plane so that their edges do not cross each other. A graph drawn in such a way is called a \emph{plane} graph. If $G$ is a plane graph, its \emph{dual} $G^*$ is a graph that has a vertex corresponding to each face of $G$, and an edge joining the vertices corresponding to neighboring faces for each edge of $G$.
Note that if $G$ is connected, $G=(G^*)^*$. The \emph{weak dual} of $G$ is the subgraph of $G^*$ whose vertices correspond to the bounded faces of $G$.

We close this section by introducing the two graph polynomials we will investigate in the paper. A vertex coloring of $G$ is an assignment of colors to the vertices of $G$ so that no edge is incident to vertices of the same color. A $t$-coloring of $G$ is a vertex coloring using at most $t$ colors. The \emph{chromatic polynomial} $P(G;t)$ counts the number of $t$-colorings of $G$; if the dependence on $t$ is implied in the context, this can be abbreviated as $P(G)$. 

A closely related polynomial is the flow polynomial. A nowhere-zero $t$-flow of $G$ is an assignment of values $1,2,\ldots,t-1$ to the edges of an arbitrary orientation of $G$ so that the total flow entering each vertex is congruent modulo $t$ to the total flow leaving each vertex. The \emph{flow polynomial} $F(G;t)$ counts the number of nowhere-zero $t$-flows in $G$; if the dependence on $t$ is implied, this can be abbreviated as $F(G)$. 

\section{Tools for computing chromatic and flow polynomials}

Before we present our main results, we will list a number of well-known facts frequently used in the computation of chromatic and flow polynomials. Proofs of these and other results are given by Tutte \cite{tuttebook}. In what follows, let $G=(V,E)$ be a graph with $k$ components.

\begin{equation}
\begin{split}
P(G) = P(G-e) - P(G/e) \text{ for any }e=uv\text{, where }u,v\in V\\
P(G) = P(G+e) + P(G/e) \text{ for any }e=uv\text{, where }u,v\in V
\end{split}
\end{equation} 

\begin{equation}
\text{If }G=G_1\cup G_2\text{ and }G_1\cap G_2=K_r\text{, then }P(G)=\frac{P(G_1)P(G_2)}{P(K_r)}
\end{equation}

\begin{equation}
\text{If }e\in E(G)\text{ is a multiple edge, then }P(G) = P(G-e)
\end{equation}


\noindent Many of these identities can also be used in the computation of flow polynomials, as the flow polynomial can be obtained from the chromatic polynomial in the following way: 

\begin{equation}
\text{If $G$ is planar, then } F(G)=\frac{P(G^*)}{t^k}.
\end{equation}

\noindent Furthermore, if $k>1$, the following identities allow us to find the flow and chromatic polynomials of each component separately.
\begin{equation}
\begin{split}
\text{If }G = G_1 \cup G_2\text{ and }G_1 \cap G_2=\emptyset\text{, then }P(G) = P(G_1)P(G_2),\\
\text{ and }F(G) = F(G_1)F(G_2).
\end{split}
\end{equation}

\noindent Finally, closed formulas for the chromatic polynomials of some specific graphs are known. In particular, if $K_n$ is the complete graph on $n$ vertices and $C_n$ is the cycle on $n$ vertices, 

\begin{eqnarray}
P(K_n) &=& t(t-1)\ldots(t-(n-1))\\
P(C_n) &=& (t-1)^n+(-1)^n(t-1)\\
P(T) &=& t(t-1)^{n-1} \; \text{ for any tree $T$ on $n$ vertices}
\end{eqnarray}

\begin{remark} 
Let $G=(V,E)$ be any graph, $S$ be a multiset over $V$, and $S'$ be the underlying set of $S$. By (3), $P(G_S)=P(G_{S'})$. Thus, when computing the chromatic polynomial of a generalized vertex join of graph $G_S$, we can assume without loss of generality that the multiplicity of every element in $S$ is 1. The reason we consider multisets instead of sets of vertices is because allowing certain muliple edges in a class of graphs corresponds to a larger class of dual graphs. In turn, this can lead to broader results about flow polynomials, allowing us to easily establish the flow-equivalence of certain graphs.

For instance, in the next section, we compute the chromatic polynomials of generalized vertex joins of trees. We show in Section 5 that the duals of these graphs are outerplanar graphs, where the added vertex $v^*$ is the one corresponding to the outer face. Allowing multiple edges between $v^*$ and each vertex of the tree means the family of duals includes all biconnected outerplanar graphs, instead of ones for which at most one edge from each bounded face borders the outer face. Thus, we are able to state a broader result about flow polynomials. A similar principle is used in Section 6 with ``generalized wheel" graphs.
\end{remark}

We are now equipped with the motivation and technical tools necessary to derive our main results. 

\section{Chromatic polynomials of generalized vertex join trees}

Let $T=(V,E)$ be a tree with $|V|=n$, $S$ be a multiset over $V$, and let $T_S$ be the generalized vertex join of $T$ using $S$. 
For short, we will call $T_S$ a \emph{generalized vertex join tree}. In this section, we will present an efficient algorithm to find $P(T_S)$.



First, by Remark 1, we can assume that the multiplicity of every element in $S$ is 1.
Two special cases of $T_S$ occur when $S=\emptyset$ and when $|S|=1$. In the first case, $T_S$ consists of a tree on $n$ vertices and an isolated vertex. Thus, by (8) and (5), we find that $P(T_S)=t^2(t-1)^{n-1}$. In the second case, $T_S$ is a tree on $n+1$ vertices, so (8) can be applied to find that $P(T_S)=t(t-1)^n$. Thus, from now on, we will assume that $|S|\geq 2$.

Next, suppose there are $b$ bridges in $T_S$, and let $B$ be the set of vertices in $T_S$ which are an endpoint of some bridge, but do not belong to a cycle. Note that since $|S|\geq 2$, there is at least one cycle, so not all edges of $T_S$ are bridges.
Let $T_S'=T_S-B$. Using (2) and (8), it is easy to see that 
\begin{equation}
P(T_S)=P(T_S')\frac{\bigl(t(t-1)\bigr)^b}{t^b}=P(T_S')(t-1)^b.
\end{equation}

We now select some vertex $r\neq v^*$ in $T_S'$. Define $L:V(T_S')\backslash \{v^*\} \rightarrow \mathbb{N}\cup \{0\}$ by $L(x)=$ (length of the shortest path between $v$ and $r$ in $T_S'-v^*$). Also, define $\mathcal{L}=\max\{L(x):x\in V(T_S')\backslash\{v^*\}\}$, and $L_i(T)=\{v:L(v)=i\}$. If $L(v)=i$, $w$ is a \emph{child} of $v$ if $w$ is adjacent to $v$ and $L(w)=i+1$; $u$ is a \emph{parent} of $v$ if $u$ is adjacent to $v$ and $L(u)=i-1$. Vertex $z$ is a \emph{descendant} of $v$ if $z=v^*$ or if there is a path $v,p_1,\ldots, p_r,z$ such that $L(v)>L(p_1)>\ldots>L(p_r)>L(z)$. The set of all descendants of $v$ is denoted $D(v)$\footnote{These definitions are analogous to standard notions in graph theory --- like level of a node in a rooted tree, etc. --- and are slightly modified to suit our purposes.}. See Figure 2 for an illustration.

\begin{figure}
\begin{center}
\includegraphics[scale=0.3]{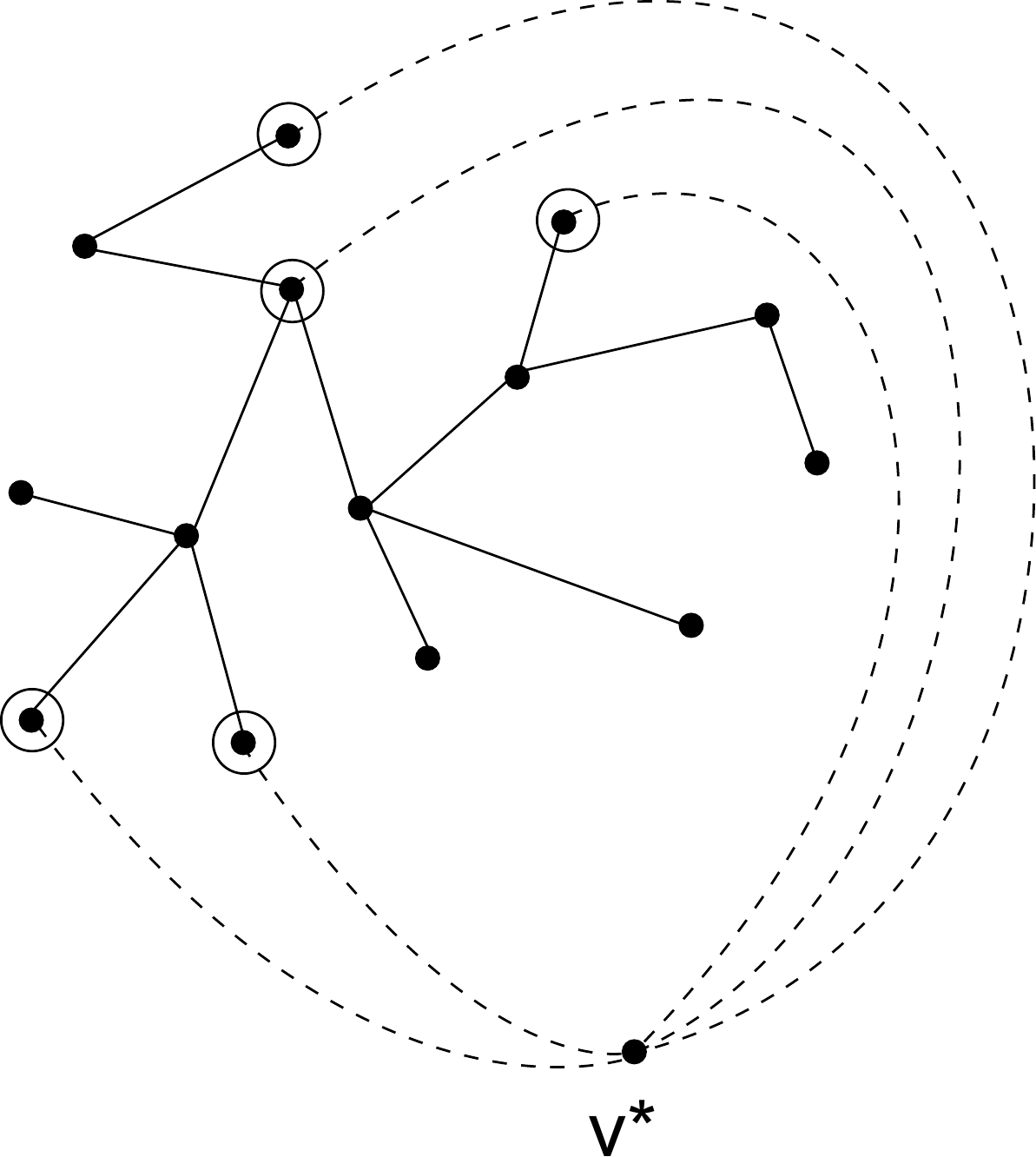}\qquad
\includegraphics[scale=0.3]{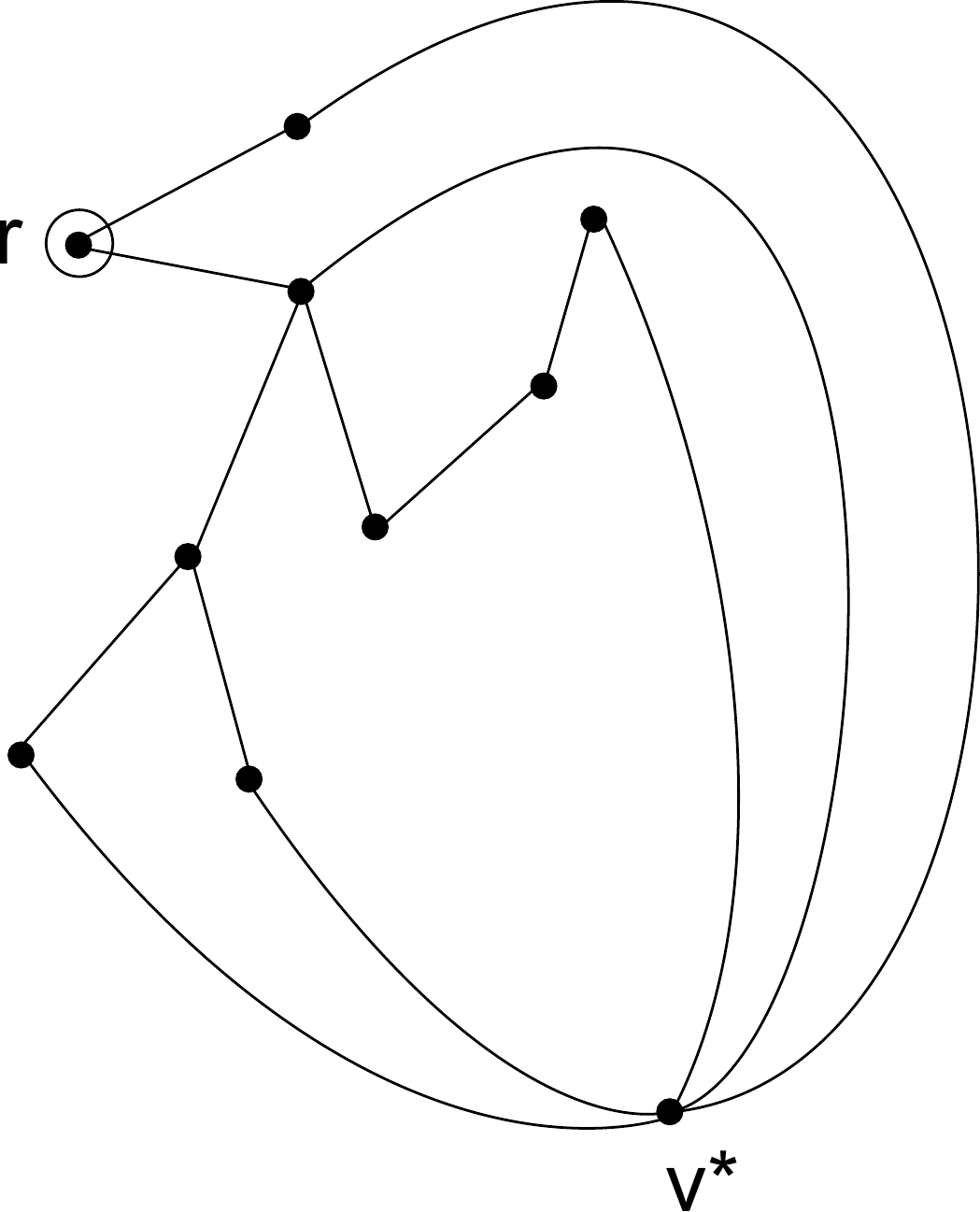}\qquad
\includegraphics[scale=0.3]{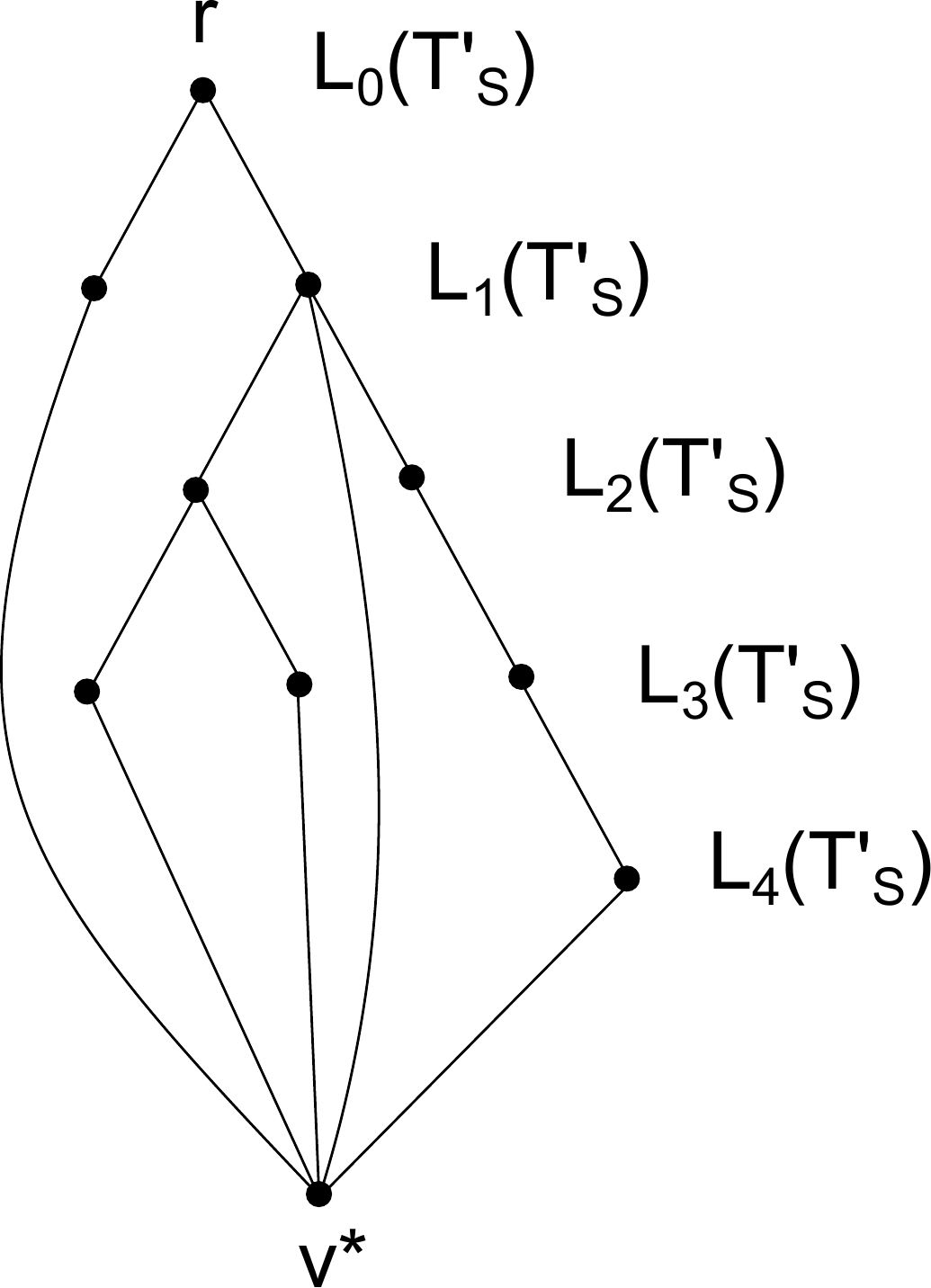}
\caption{\emph{Left}: Forming $T_S$ from a given tree $T$ and a subset of its nodes $S$. \emph{Middle}: Removing the bridges of $T_S$ to form $T_S'$, and selecting a node $r$. \emph{Right}: Finding $L(v)$ and $L_i(T_S')$.}
\end{center}
\end{figure}

Next, define the function $f:V(T_S')\backslash \{v^*\} \rightarrow \{0,1\}$ by $f(x)=1$ if $x\in S$, $f(x)=0$ if $x\notin S$. 
Finally, for any $a\in V(T_S')\backslash \{v^*\}$, we define $T_a=T_S'[a\cup D(a)]$, $H_a=T_a/rv^*$, $\tilde{T}_{c}=T_S'[\{a\}\cup \{c\} \cup D(c)]$, and 
$\tilde{H}_{c}=\tilde{T}_c/av^*$,
where $c$ is a child of $a$. See Figure 3 for an illustration of these definitions.

\begin{figure}
\begin{center}
\includegraphics[scale=0.3]{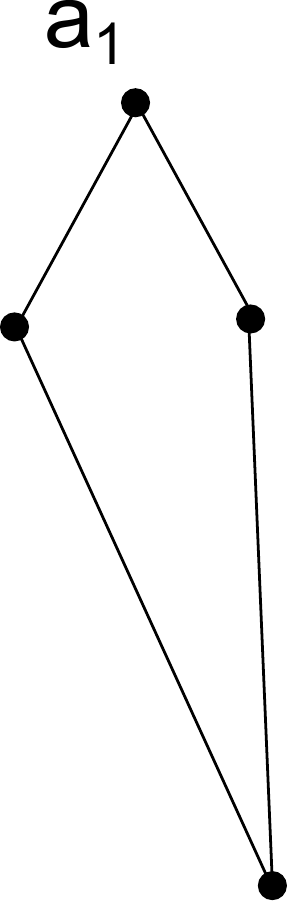}\qquad \qquad
\includegraphics[scale=0.3]{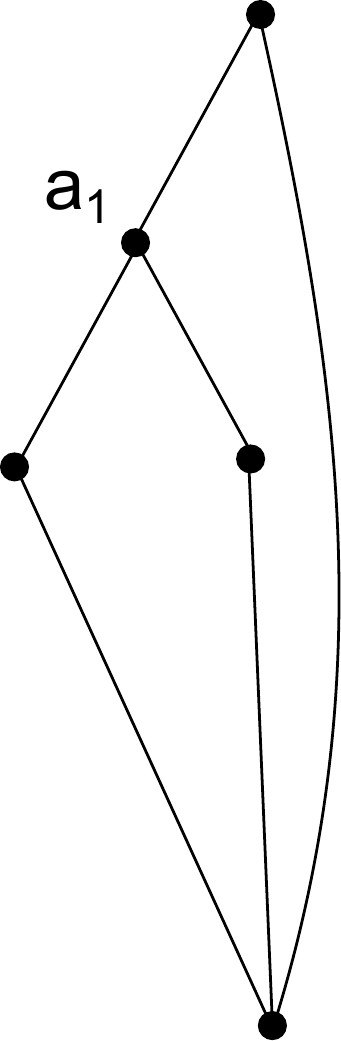}\qquad \qquad
\includegraphics[scale=0.3]{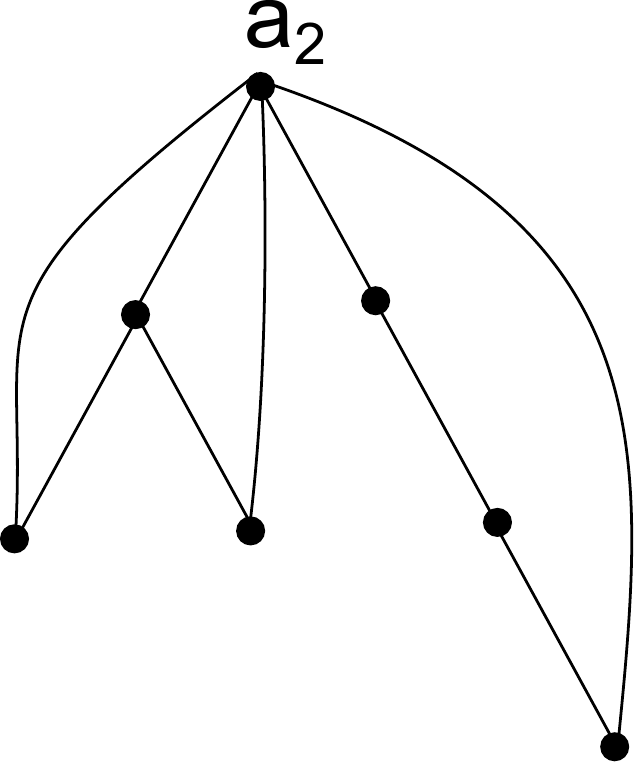}\qquad \qquad
\includegraphics[scale=0.3]{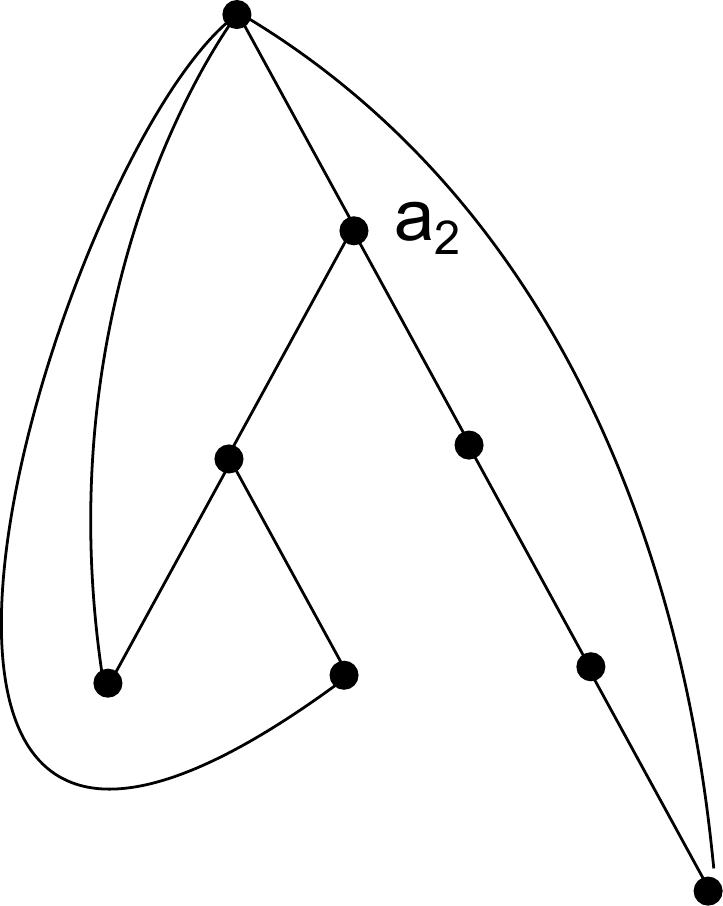}
\caption{\emph{From left to right:} $T_{a_1}$; $\tilde{T}_{a_1}$; $H_{a_2}$; $\tilde{H}_{a_2}$, for two vertices $a_1$ and $a_2$ of the graph $T_S'$ shown in Figure 2, right.}
\end{center}
\end{figure}

Let $a\neq v^*$ be a vertex with children $c_1,\ldots,c_k$, and suppose we know  $P(T_{c_i})$ and $P(H_{c_i})$, $1\leq i \leq k$.
Let $I=\{i : f(c_i)=1\}$, $Z=\{i : f(c_i)=0\}$. Thus, $I$ and $Z$ form a partition of $\{1,\ldots,k\}$.
Then, we can calculate $P(H_a)$ as follows.

\begin{eqnarray*}
P(H_a)&=&
\frac{1}{t^{k-1}}
\prod_{i=1}^k P(\tilde{H}_{c_i})\\
&=& 
\frac{1}{t^{k-1}}
\prod_I P(\tilde{H}_{c_i})
\prod_Z P(\tilde{H}_{c_i})\\ 
&=&
\frac{1}{t^{k-1}}
\prod_I P(T_{c_i})
\prod_Z \bigl(P(T_{c_i})-P(H_{c_i})\bigr)
\end{eqnarray*}

\noindent Here, the first equality follows from (2) and the third equality follows from (1).

Next, we will compute $P(T_a)$ by considering two cases: $a$ is either in $S$ or not. Let $P_1(T_a)=P(T_a)$, where $f(a)=1$, and $P_0(T_a)=P(T_a)$, where $f(a)=0$.
Clearly, $P(T_a)=f(a)P_1(T_a)+\bigl(1-f(a)\bigr)P_0(T_a)$. We now find $P_1(T_a)$ and $P_0(T_a)$ separately as follows.

\begin{eqnarray*}
P_1(T_a)&=&
\frac{1}{\bigl(t(t-1)\bigr)^{k-1}}
\prod_{i=1}^k P(\tilde{T}_{c_i})\\
&=&
\frac{1}{\bigl(t(t-1)\bigr)^{k-1}}
\prod_I P(\tilde{T}_{c_i})
\prod_Z P(\tilde{T}_{c_i})\\
&=&
\frac{\prod_I\bigl(P(T_{c_i})(t-2)\bigr)}{\bigl(t(t-1)\bigr)^{k-1}}
\prod_Z\bigl(P(\tilde{T}_{c_i}+c_iv^*)+P(\tilde{T}_{c_i}/c_iv^*)\bigr)\\
&=&
\frac{\prod_I\bigl(P(T_{c_i})(t-2)\bigr)}{\bigl(t(t-1)\bigr)^{k-1}}
\prod_Z\bigl(P(T_{c_i}+c_iv^*)(t-2)+P(H_{c_i})(t-1)\bigr)\\
&=&
\frac{\prod_I\bigl(P(T_{c_i})(t-2)\bigr)}{\bigl(t(t-1)\bigr)^{k-1}}
\prod_Z\Bigl(\bigl(P(T_{c_i})-P(H_{c_i})\bigr)(t-2)+P(H_{c_i})(t-1)\Bigr)\\
&=&
\frac{1}{\bigl(t(t-1)\bigr)^{k-1}}
\prod_I\bigl(P(T_{c_i})(t-2)\bigr)
\prod_Z\bigl((t-2)P(T_{c_i})+P(H_{c_i})\bigr)\\
P_0(T_a)&=&P(T_a+av^*)+P(T_a/av^*)=P_1(T_a)+P(H_a)
\end{eqnarray*}

\noindent Here, in the computation of $P_1(T_a)$, the first equality follows from (2), the third equality follows from (2) and (1), the fourth equality follows from (2), and the fifth equality follows from (1). In the computation of $P_0(T_a)$, the first equality follows from (1), and the second follows from the definitions of $P_1$ and $H_a$.

Thus, we have shown how to express $P(T_a)$ and $P(H_a)$ in terms of $P(T_{c_i})$ and $P(H_{c_i})$, $1\leq i \leq k$. Using these identities, we propose the following algorithm for finding the chromatic polynomial of a generalized vertex join tree $T_S$.

\begin{framed}
\textbf{Algorithm 1} \par
\begin{enumerate}
\item Find and remove the bridges of $T_S$ to acquire $T_S'$ \par
\item For $i=\mathcal{L}$ to $0$ \par
\qquad Compute $P(T_a)$ and $P(H_a)$ for each $a\in L_i(T_S')$ \par
\item Compute $P(T_S)$ using (9)
\end{enumerate}
\end{framed}

\begin{theorem}
Algorithm 1 finds the correct chromatic polynomial of a generalized vertex join tree $T_S$ using polynomial time and space.
\end{theorem}

\begin{proof}
We already showed that by finding the bridges of $T_S$ and the chromatic polynomial of $T_S'$, we can compute $P(T_S)$. Thus, we only need to show that Step~2 of the algorithm correctly computes $P(T_S')$.

We also established that we can compute $P(T_a)$ and $P(H_a)$ for any $a \in T_S'$, $a\neq v^*$, if we know $P(T_c)$ and $P(H_c)$ for every child $c$ of $a$. Note that this condition is automatically satisfied for vertices which have no children.

By construction, vertices in $L_{\mathcal{L}}(T_S')$ have no children, so $P(T_a)$ and $P(H_a)$ can be found immediately for any vertex $a\in L_{\mathcal{L}}(T_S')$. For $\mathcal{L}>i\geq 0$, a vertex $a$ in $L_i(T_S')$ either has no children, or has all of its children in $L_{i+1}(T_S')$. In either case, $P(T_c)$ and $P(H_c)$ are known for every child $c$ of $a$ -- either vacuously or inductively. Thus, $P(T_a)$ and $P(H_a)$ can also be computed. Since by construction, $P(T_S')=P(T_r)$ and $L_0(T_S')=\{r\}$, Algorithm 1 indeed finds the correct chromatic polynomial of $T_S$.

To verify the time- and space-complexity of the algorithm, let $|V(T)|=n$. By construction, $|E(T_S)|=\mathcal{O}(n)$. Thus, using the algorithm of Tarjan \cite{tarjan}, we can find all bridges of $T_S$ in $\mathcal{O}(n)$ time.

Next, we can create a tree data structure with root $r$ using an adjacency-preserving bijection between the nodes of the data structure and the vertices of $T_S'-v^*$. This will require $\mathcal{O}(n)$ time and space due to the sparsity of the graph. To evaluate the function $L$ at each vertex of $T_S'-v^*$, we find the level of the corresponding node in the data structure, which can be done with $\mathcal{O}(n)$ total time. From this data structure, we can easily find $\mathcal{L}$ and $L_i(T_S')$, $\mathcal{L}\geq i \geq 0$, as well as the parent, children, and descendants of each node. 

Finally, each evaluation of $P(T_a)$ and $P(H_a)$ requires the multiplication of $\mathcal{O}(a_k)$ polynomials, where $a_k$ is the number of children of $a$. 
Since we evaluate $P(T_a)$ and $P(H_a)$ for $\mathcal{O}(n)$ vertices, and the total number of children in $G_S$ is $\mathcal{O}(n)$, the evaluation of $P(T_S)$ requires the multiplication of $\mathcal{O}(n)$ polynomials. Each of these polynomials has degree $\mathcal{O}(n)$, since $P(T_S)$ has degree $\mathcal{O}(n)$. The time-complexity of multiplying two polynomials of degree $n$, using a Fast Fourier Transform, is $\mathcal{O}(n\log n)$, so the total time complexity of Algorithm 1 is $\mathcal{O}(n^2\log n)$.

A polynomial of degree $n$ can be stored with $\mathcal{O}(n)$ space, and we have to store no more than $2n$ polynomials at a time, in addition to the auxiliary tree data structure. Thus, the total space-complexity is $\mathcal{O}(n^2)$.
\qed
\end{proof}

\section{Flow polynomials of outerplanar graphs}

Let $G$ be a simple biconnected outerplane graph with bounded faces $F_1,\ldots, F_s$ and outer face $F_*$. The weak dual of $G$ is a tree $T=(V,E)$, where $v_i\in V$ corresponds to face $F_i$ of $G$. Suppose $F_i$ shares $f_i$ edges with $F_*$. If we let $v^*$ correspond to $F_*$ and $S$ be the multiset over $V$ where $v_i$ appears $f_i$ times, then $T_S$ is the dual of $G$. The following lemma implies that the converse is also true.


\begin{lemma}
A generalized vertex join tree $T_S$ is the dual of a simple biconnected outerplanar graph $G$ if and only if every vertex of $T_S$ has degree at least 3.
\end{lemma}

\begin{proof}

We have already shown that if $G$ is a simple biconnected outerplanar graph, the dual of $G$ is a generalized vertex join tree $T_S$. If $T_S$ has a vertex $v_0$ with degree 0, the face of $G$ corresponding to $v_0$ has no edges, meaning $G$ has only one face (the outer face), so $G$ consists of just one vertex $u_0$. If $T_S$ has a vertex $v_1$ with degree 1, the face of $G$ corresponding to $v_1$ has one edge, meaning it is a loop. If $T_S$ has a vertex $v_2$ with degree 2, the face of $G$ corresponding to $v_2$ has two edges, meaning they are parallel edges. All three cases contradict $G$ being simple and biconnected. Thus, if $T_S$ is the dual of a simple biconnected outerplanar graph, every vertex of $T_S$ has degree at least 3.

Now, suppose $T_S$ is a generalized vertex join tree, and that every vertex of $T_S$ has degree at least 3. We will show that $T_S$ is the dual of a simple biconnected outerplanar graph by induction on the number of vertices of $T_S$. If $T_S$ has 2 vertices $v$ and $v^*$, all the edges in $T_S$ must join $v$ to $v^*$ since by construction, $T_S$ can have no loops. Thus, $T_S$ is the dual of some cycle of size at least 3 (which is simple, biconnected, and outerplanar). Next, let $T_S$ be a generalized vertex join tree on $k+1$ vertices with minimum vertex degree at least 3, and let $v$ be a leaf of $T$. Since $T$ is a tree, $v$ has a unique neighbor $u$ in $T$ and there is exactly one edge between $u$ and $v$. Moreover, by assumption, $v$ is connected to $v^*$ by $\ell\geq 2$ edges and $u$ is incident to at least two edges other than $uv$. Thus, if we delete $v$ from $T_S$ and add an edge from $u$ to $v^*$, we obtain a generalized vertex join tree on $k$ vertices, which by induction is the dual of some simple biconnected outerplanar graph $G$. In this graph, $u$ corresponds to some bounded face $F$ and $v^*$ corresponds to the outer face $F_*$. Since we added an edge $uv^*$, $F$ shares at least one edge $e$ with $F_*$. Now, if we glue a cycle of size $\ell-1$ to $e$, we obtain a simple biconnected outerplanar graph whose dual is $T_S$. \qed

\end{proof}

From the above arguments, the following lemma easily follows. Note that a biconnected outerplanar graph with loops has an embedding in which no loop edge borders the outer face, and therefore the weak dual graph is a tree and the (strong) dual graph is a generalized vertex join tree. Since chromatic and flow polynomials are independent of embedding, we can consider such an embedding without loss of generality.

\begin{lemma}
Let $G$ be a biconnected outerplanar graph, and $T_S$ be its dual generalized vertex join tree. $G$ has multiple edges if and only if $T_S$ has degree 2 vertices, and $G$ has loops if and only if $T_S$ has degree 1 vertices.
\end{lemma}

\noindent Thus, the dual of any (not necessarily simple) biconnected outerplanar graph is a generalized vertex join tree $T_S$. Now, let $G$ be any outerplanar graph. If $G$ has several components, by (5) we can compute the flow polynomial of each component separately. If a connected outerplanar graph is not biconnected, then it has a bridge, so its flow polynomial is zero. Thus, we can compute the flow polynomial of any outerplanar graph using the algorithm from Section 4; we state this in the corollary below.

\begin{corollary}
The flow polynomial of a (not necessarily simple) outerplanar graph $G$ can be found in polynomial time and space using Algorithm 1.
\end{corollary}
\begin{proof}
If $|V(G)|=n$, then since $G$ is outerplanar, $|E(G)|=\mathcal{O}(n)$. We can start by taking $\mathcal{O}(n)$ time to find the connected components and detect any bridges in $G$, and henceforth suppose without loss of generality that $G$ is biconnected. Then, $G^*$ can also be found in $\mathcal{O}(n)$ time by first finding the faces of $G$ via a breadth-first traversal and then forming the necessary vertices and adjacencies for $G^*$. Finally, Algorithm 1 can be applied to find $P(G^*)$, and $F(G)$ can be found using (4). Thus, the total time complexity of this procedure is dominated by the time-complexity of Algorithm 1 and is therefore $\mathcal{O}(n^2 \log n)$. The space-complexity is equal to the space complexity of Algorithm 1, plus the space necessary to store the dual of $G$. Thus, the total space complexity for finding $F(G)$ is $\mathcal{O}(n^2)$.
\qed
\end{proof}


\section{Generalized vertex joins of cliques and cycles}

In this section, we compute the chromatic polynomials of generalized vertex joins of complete graphs and cycles, and use duality to compute the flow polynomials of generalized vertex joins of cycles.

\subsection{Chromatic polynomials of generalized vertex join cliques}

Let $K=(V,E)$ be a complete graph, $S$ be a multiset over $V$ and let $K_S$ be the generalized vertex join of $K$ using $S$. For short, we will call $K_S$ a \emph{generalized vertex join clique}. Let $|V|=n$, $S'$ be the underlying set of $S$ and $|S'|=s$. Using (3), (2), and (6), we have 

\begin{equation*}
P(G_S)= P(G_{S'})=\frac{P(K_n)P(K_{s+1})}{P(K_s)}=(t-s)\prod_{i=0}^{n-1} (t-i).
\end{equation*}

\noindent Since in general complete graphs are not planar, in this case there is no useful dual result about flow polynomials.

\subsection{Chromatic polynomials of generalized vertex join cycles}

Let $C=(V,E)$ be a cycle, $S$ be a multiset over $V$ and let $C_S$ be the generalized vertex join of $C$ using $S$. For short, we will call $C_S$ a \emph{generalized vertex join cycle}. 
In the literature, graphs of this form have also been called ``generalized wheel" or ``broken wheel" graphs, and have been investigated in other contexts (cf. \cite{brokenwheel}). In the remainder of this section, we will present formulas for $P(C_S)$ and $F(C_S)$. In doing so, we will avoid burdening ourselves with unnecessary formalization, as most of the ideas presented here are fairly intuitive. Instead, whenever possible, we refer the reader to figures in the hope that they will be sufficient to dispel ambiguity.

Let us think of $C_S$ as a plane graph with a standard ``wheel" embedding (see Figure 4).
In this embedding of $C_S$, we will call the edges incident to $v^*$ \emph{spokes} and the edges bordering the outer face \emph{outer edges}. It is easy to see that if $|V|=n\geq 1$ and $|S|=s\geq 1$, $C_S$ has $n$ outer edges, $s$ spokes, and $s$ bounded faces. 
Let $v_1,\ldots,v_n$ be the vertices of $V$ in clockwise order, and let $\phi(C_S)=a_1,\ldots,a_n$ be a string where $a_i$ is the multiplicity of $v_i$ in $S$. This assignment produces a bijection $\phi$ between the set of generalized vertex join cycles and the set of nonnegative integer strings. Furthermore, it is easy to see that the dual of $C_S$ is also a generalized vertex join cycle with $n$ spokes, $n$ bounded faces, and $s$ outer edges, and that $\phi(C_S^*)$ can easily be obtained from $\phi(C_S)$. See Figure 4 and (11) for details.

\begin{figure}
\begin{center}
\includegraphics[scale=0.5]{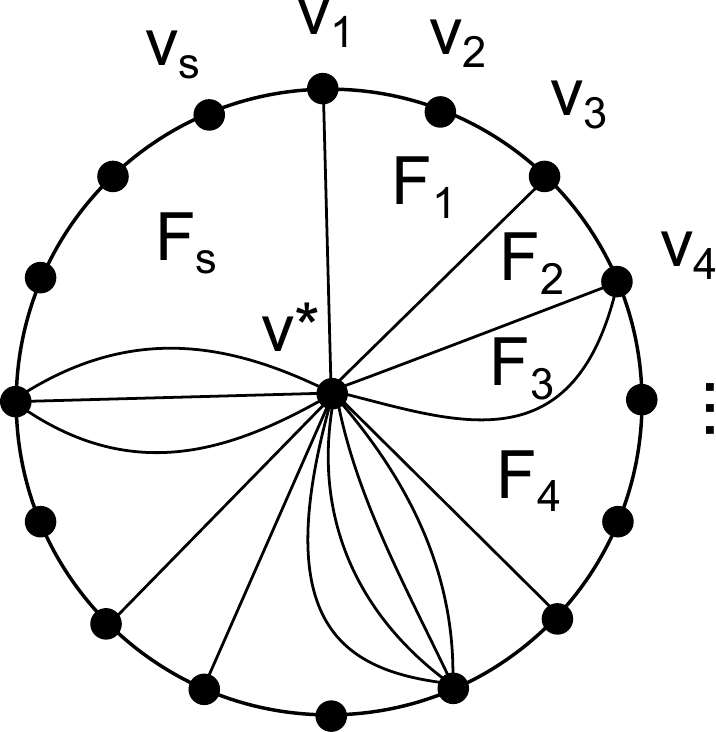}\qquad \qquad
\includegraphics[scale=0.5]{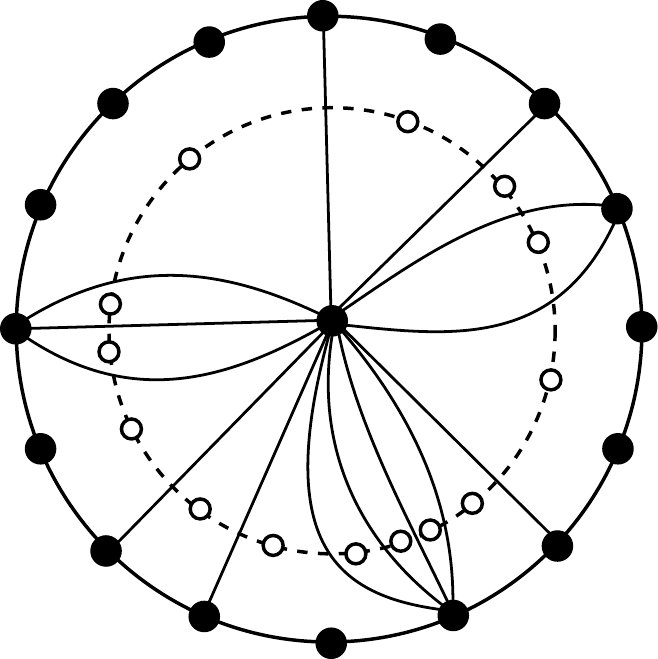}\qquad \qquad
\includegraphics[scale=0.55]{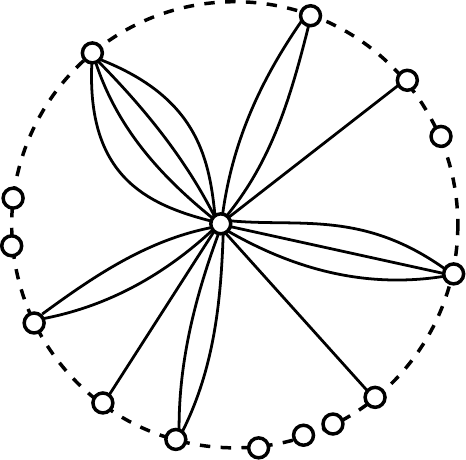}\qquad \caption{\emph{Left:} Generalized vertex join cycle $C_S$.  
\emph{Middle:} Obtaining $C_S^*$ from $C_S$.
\emph{Right:} $C_S^*$, the dual of $C_S$; note that $C_S^*$ is also a generalized vertex join cycle.
\newline
$\phi(C_S)=1,0,1,2,0,0,1,4,0,1,1,0,3,0,0,0$; 
$\phi(C_S^*)=2,1,0,3,1,0,0,0,2,1,2,0,0,4$.}
\end{center}
\end{figure}




Now, suppose $S'$ is the underlying set of $S$; by (3), 
$P(C_S)=P(C_{S'})$.
\noindent Also, $F(C_S)=\frac{1}{t}P(C_S^*)$. But $C_S^*$ is some generalized vertex join cycle $\tilde{C}_{\tilde{S}}$, so by (4) and (3), 
\begin{equation*}
F(C_S)=\frac{1}{t}P(C_S^*)=\frac{1}{t}P(\tilde{C}_{\tilde{S}})=\frac{1}{t}P(\tilde{C}_{\tilde{S'}}),
\end{equation*}
\noindent where $\tilde{S'}$ is the underlying set of $\tilde{S}$. Thus, if we can find the chromatic polynomial of a generalized vertex join cycle $C_{S'}$ where every element of $S'$ has multiplicity 1, we can find the chromatic and flow polynomials of all generalized vertex join cycles.

With this in mind, without loss of generality, let $C_S$ be a generalized vertex join cycle where every element of $S$ has multiplicity 1. Let $e_1,\ldots,e_s$ be the spokes of $C_S$ in clockwise order, and $F_1,\ldots,F_s$ be the corresponding faces of $C_S$ (as labeled in Figure 4). Define $C_S^i=C_S-\{e_i,\ldots,e_s\}$ and $C_S^{s+1}=C_S$. Then, repeatedly using (1), we have the following identity; see Figure 5 for an illustration.


\begin{eqnarray}
\nonumber P(C_S)&=&P(C_S-e_s)-P(C_S/e_s)\\
\nonumber &=&P(C_S^s)-P(C_S^{s+1}/e_s)\\
&=&P(C_S^{s-1})-P(C_S^s/e_{s-1})-P(C_S^{s+1}/e_s)\\
\nonumber &\vdots &\\ 
\nonumber &=&P(C_S^1)-P(C_S^{2}/e_1)-\ldots -P(C_S^s/e_{s-1})-P(C_S^{s+1}/e_s)\\
\nonumber &=&tP(C_n)-\sum_{i=1}^s P(C_S^{i+1}/e_i).
\end{eqnarray}

\begin{figure}
\begin{center}
\includegraphics[scale=0.26]{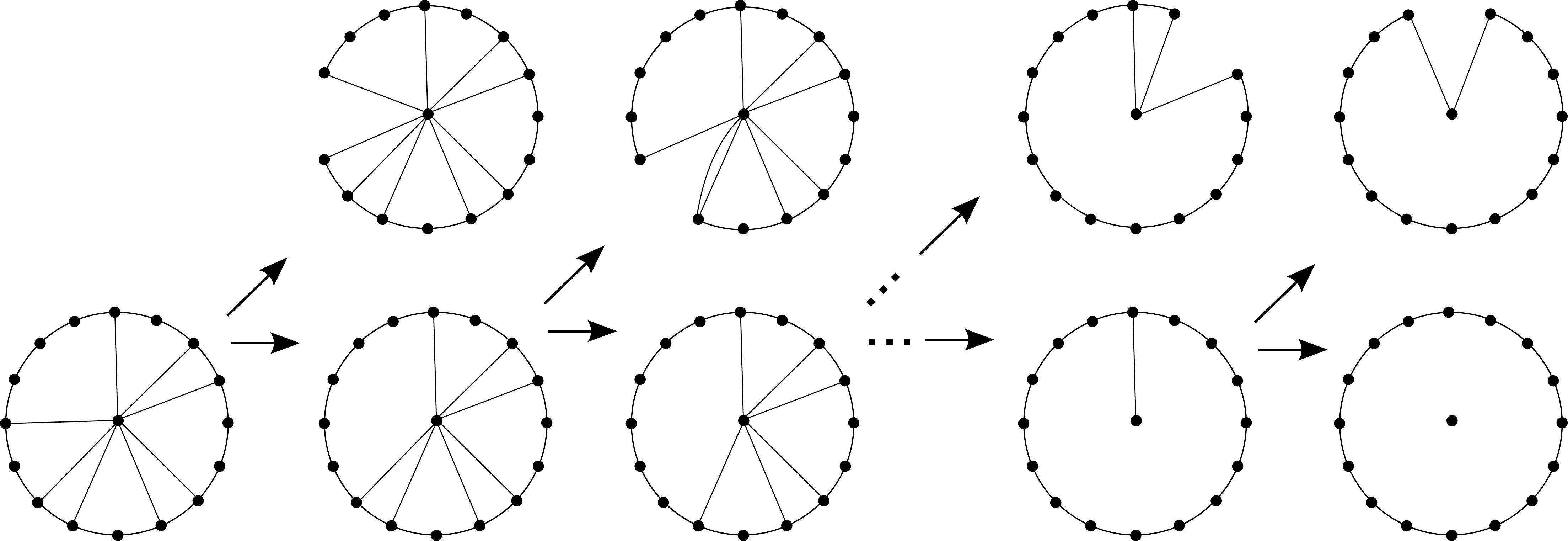}\qquad \caption{Decomposing $C_S$ (on far left) into simpler graphs as described in (10). Using (1), the graphs in the top row can be further decomposed into the cycles making up their bounded faces, since in these graphs, the faces bordering the contraction share only one edge  with the rest of the graph.}
\end{center}
\end{figure}

\noindent We can think of the faces of $C_S$ as cycles, where face $F_i$ has size $f_i$. Then, the faces of $C_S^i$ have sizes $f_1,\ldots, f_{i-2}, u_i$, where 

\begin{eqnarray*}
u_i&=&2+(f_{i-1}-2)+\ldots,+(f_s-2)\\
&=&2+(-2)(s-i+2)+\sum_{j=i-1}^sf_j\\
&=&2(i-s-1)+\sum_{j=i-1}^s f_j.
\end{eqnarray*}

\noindent Hence, the faces of $C_S^i/e_{i-1}$ have sizes $f_1,\ldots, f_{i-2}-1, u_i-1$. Let $G_i$ be the (multi)set of sizes of faces of $C_S^i/e_{i-1}$, i.e. $G_i=\{f_1,\ldots, f_{i-2}-1, u_i-1\}$. Then, we can start from a face of $C_S^i/e_{i-1}$ which borders the contracted edge, and use the fact that this face shares just one edge with the rest of the graph to successively apply (2) and evaluate $P(C_S^i/e_{i-1})$. Then, $P(C_S^{i-1}/e_i)=P(K_2)^{1-i}\prod_{j\in G_i}P(C_j)$, so

\begin{eqnarray*}
P(C_S)&=&tP(C_n)-\sum_{i=1}^s \frac{\prod_{j\in G_i}P(C_j)}{P(K_2)^{i-1}}\\
&=&t((t-1)^n+(-1)^n(t-1))-\sum_{i=1}^s \frac{\prod_{j\in G_i}((t-1)^j+(-1)^j(t-1))}{(t(t-1))^{i-1}}.
\end{eqnarray*}

Now, as was mentioned previously, a generalized vertex join cycle $C_S$ corresponds uniquely to an integer string $\phi(C_S)$. From $\phi(C_S)=a_1,\ldots,a_n$, we can obtain $\phi(C_S^*)$ by simple and fast manipulations. For instance, $\phi(C_S^*)$ can be constructed by starting from an empty string, and for $1\leq i\leq n$, concatenating:

\begin{flalign}
\begin{split}
a_i-1 \text{ zeros}&\text{ if }a_i \geq 2,\\
1 &\text{ if } a_i = 1,\\
z+1 &\text{ if } a_i \text{ is the beginning of a maximal connected }\\
&\qquad \text{substring of zeros of } \phi(C_S) \text{ with length } z.
\end{split}
\end{flalign}


\noindent Moreover, the sizes of the faces of $C_S$, $f_1,\ldots,f_s$, can be obtained by adding 2 to every element of $\phi(C_S^*)$. Thus, these definitions, along with (11), can be used to compute $G_i$ used in the formula above, and thus we can compute $P(C_S)$ and $F(C_S)$ for any generalized vertex join cycle $C_S$.

\section{Conclusion}

We have found a low-order polynomial time algorithm for computing the chromatic polynomials of generalized vertex join trees. This algorithm can also be used to find the flow polynomials of outerplanar graphs. We also derived closed formulas for the chromatic polynomials of generalized vertex join cliques, and the chromatic and flow polynomials of generalized vertex join cycles. Future work will focus on finding the chromatic polynomials of generalized vertex joins of other families of graphs, and of graphs which have several generalized vertex joins.

\section{Acknowledgements}
This material is based upon work supported by the National Science Foundation under Grant No. 1450681.

\end{document}